\documentclass[12pt,english]{article}
 
\usepackage{centernot}
\usepackage{amsfonts}
\usepackage{amssymb}
\usepackage{amsmath}
\usepackage{amsfonts}
\usepackage{amssymb}
\usepackage{color}

\usepackage{enumitem}
\usepackage{tikz}
\usetikzlibrary{positioning,chains,fit,shapes,calc}
\usetikzlibrary{matrix,arrows.meta}
\usepackage{graphicx}
\usepackage{booktabs}
\usepackage{tabularx}
\usepackage{amsmath}
\usepackage{array}
\usepackage[margin=3cm]{geometry}
\usepackage{comment}
\usepackage{amsthm}
\usepackage[mathscr]{euscript}
\newcolumntype{m}[1]{>{\centering\arraybackslash}p{#1}}
\usepackage{adjustbox}
\usepackage{flafter}
\usepackage{placeins}
\usepackage{bm}

\newtheorem{observation}{Observation}
\newtheorem{lemma}{Lemma}

\newtheorem{theorem}{Theorem}

\newtheorem{example}{Example}

\newtheorem{proposition}{Proposition}

\usepackage[ruled,noend]{algorithm2e}

\usepackage{natbib}
\usepackage{setspace}

\usepackage{todonotes} 
\usepackage{comment} 

\newif\ifshow 
\showfalse 

\ifshow
  \includecomment{wrap}
\else
  \excludecomment{wrap} 
\fi

\begin{document}

\title{\bf Strategyproof Maximum Matching under Dichotomous Agent Preferences
	\thanks{The authors thank Florian Brandl, Lars Ehlers, Sean Horan, Pinaki Mandal, Vikram Manjunath, Shivika Narang, Adrian Vetta, and audience members at the 2023 Interdisciplinary CIREQ Workshop in Computer Science and Economics for their comments.}}

\author{Haris Aziz\thanks{UNSW Sydney, Sydney, Australia; e-mail:haris.aziz@unsw.edu.au} \, \,
Md. Shahidul Islam\thanks{Concordia University, Montreal, Canada; e-mail:  md-shahidul.islam@mail.concordia.ca} \, \,
Szilvia P\'{a}pai\thanks{Concordia University, Montreal, Canada; e-mail: szilvia.papai@concordia.ca} }

\date{}


\maketitle

\begin{abstract}

\linespread{1.3}\selectfont

We consider a two-sided matching problem  in which the agents on one side have dichotomous preferences and the other side representing institutions has strict preferences (or priorities).  It captures several important applications  in matching market design in which the agents are only interested in getting matched to an acceptable institution. These include centralized daycare assignment and healthcare rationing. We present a compelling new mechanism that satisfies many prominent and desirable properties including individual rationality, maximum size, fairness, Pareto-efficiency on both sides, strategyproofness on both sides, non-bossiness and having polynomial time running time. As a result, we answer an open problem whether  there exists a mechanism that is agent-strategyproof, maximum, fair and non-bossy.

\medskip


\medskip

 
\end{abstract}

\section{Introduction}
\label{sec:introduction}

%
%
%
%



We consider a fundamental matching problem in which agents are to be matched to institutions. 
Agents express dichotomous preferences over institutions by specifying acceptable institutions. Institutions express preferences  among the set of agents but cannot declare any agent unacceptable. There is a shortage of the available spots at institutions, and agents care primarily about being assigned to an acceptable object rather than which institution they are assigned to. The system typically requires that agents are matched to institutions in a way that is fair with respect to the preferences of both sides. At the same time, the designer is concerned about assigning as many agents to acceptable spots as possible.  

The problem that we consider is inspired by the centralized matching of children to daycare centres or schools. In these problems, parents of children express a subset of daycares as acceptability if they are near enough or provide the required facilities~(see, e.g., \citep{STM+23a}). On the other hand, each daycare may have its own preference ordering over children according to its own criteria.  Such a problem also arise with an increasing number of applicants to various institutions, such as when there are too few school seats compared to the number of school-age children, or when there is a shortage of daycare capacity compared to the demand for daycare. Accommodating applicants in schools, daycares, and similar institutions is a frequent concern in many places worldwide, where applicants may stay on waiting lists for a long time. 
Depending on the application, the spots at the institution could be daycare spots,  immigration slots, school seats, or healthcare treatments. In particular, our setting  captures healthcare rationing problems in which healthcare patients are matched to limited healthcare resources. These problems have received tremendous interest in recent years~(see,. e.g., \citep{,PSU+23a,AzBr24a}).  


In the centralized setting that we consider, our main goals is to maximize the number of placements of agents to acceptable slots and to do this in a fair and strategyproof manner.  
There are many decentralized daycare systems with a substantial shortage of daycare openings.  There is often a major concern that many children are on the waiting lists of daycares for an extended period of time, and thus parents face difficulties when returning to work after parental leave. Surprisingly, there are still vacant spots remaining in daycares which implies inefficient usage of resources. However, matching the maximum number of applicants is unlikely under a decentralized allocation system. According to \citet{ChKo16a}, decentralized matching mechanisms lead to unfairness and inefficiency. These might be unexpected consequences if applicants remain unassigned due to problems in the allocation system. Matching as many children as possible to daycares, therefore, is the primary objective. Secondarily, we also want to ensure that the allocation is fair, which means that applicants’ preferences are respected. Finally, it is also important to use an allocation mechanism that provides the correct incentives when applicants report their preferences, in order to make sure that the allocation is indeed efficient and maximizes the number of matched applicants, and also respects the preferences over applicants. Many of the above concerns are captured by key axioms in market design that we discuss below.

When finding a desirable matching, we are guided by basic axiomatic properties that are well-established in the matching literature. 
A basic requirement is to find an \textit{individually rational} matching in which each agent is matched to an acceptable institution and each institution is matched to an acceptable agent. A matching satisfies \textit{fairness} (respects the preferences/priorities of institutions) if there is no unmatched agent $i$ who wishes to be matched to an acceptable institution $d$ and $d$ is matched to an agent $j$ that is ranked lower than $i$ by $d$.  
We want to find a maximum size individually rational matching that satisfies respect for priorities. 
Interpreted in our setting, a result of \citet{AzBr21a} implies  that a  maximum size (individually rational) matching that respects priorities can be computed in polynomial time. However, their rule still has some limitations.  It is not strategyproof in the classical sense for either of the two sides. It is also not necessarily Pareto-efficient from the institutions' side. Whether there exists a mechanism that is agent-strategyproof, maximum, fair and non-bossy was posed as an unresolved problem by \citet{AzBr21a}.  

In this paper we explore the following fundamental question:

\begin{quote}
	\emph{Can maximum size, efficiency, fairness and strategyproofness be satisfied simultaneously in a two-sided matching model with dichotomous preferences on one side and strict preferences (or priorities) on the other side? }
\end{quote}

The combination of  individual rationality, fairness, and non-wastefulness is typically called stability in classical two-sided matching problems.  When both sides have strict or weak preferences/priorities, combining fairness, efficiency, and strategyproofness yields impossibility results. It is well-known that for the classical setting with strict preferences (priorities) on both sides (1) there is no mechanism that satisfies stability and strategyproofness for both sides \citep{Roth82a}; (2) stability and Pareto-efficiency on one side are incompatible~\citep{balinski1999tale}; (3) none of the well-known stable and/or efficient mechanisms are maximum size \citep{ouguz2020assignment}. Moreover, finding a maximum size stable matching is NP-hard when there can be ties in the preferences~\citep{BMM10a}. By contrast, for the setting that we consider with dichotomous agent preferences, the answer to our fundamental research question is surprisingly positive.

\paragraph{Contributions}
We present mechanisms for two-sided matching with dichotomous preferences on one side and strict preferences on the other side that satisfy the following properties:
\begin{enumerate}[nosep]
	\item non-wastefulness;
	\item individual rationality;
	\item maximum size;
	\item fairness;
	\item Pareto-efficiency on both sides; 
	\item strategyproofness on both sides; 
	\item non-bossiness;
	\item polynomial-time.
\end{enumerate}

\medskip
To the best of our knowledge, these are the first known mechanisms that satisfy all these properties simultaneously, including the striking properties of strategyproofness and Pareto-efficiency on both sides. 
We present two families of mechanisms, the SAFE and Rank-Maximal mechanisms. SAFE mechanisms are based on the idea of safe blocks that identify subsets of institutions that are not over-subscribed by agents and should be assigned first to reach a maximum matching. We also propose another class of mechanisms with a different perspective, the Rank-Maximal mechanisms,  based on the graph-theoretic notion of rank maximality. We show that these two families of mechanisms are equivalent. We then show that these mechanisms satisfy all the key properties listed above. The two different formulations for essentially the same mechanisms has several advantages. Firstly, the multiple perspectives provide different combinatorial insights into the mechanisms and provide tools to carefully analyse them and establish axiomatic properties. Whereas Rank-Maximal mechanisms are clearly polynomial-time, SAFE mechanisms provide additional intuition about how these mechanisms work and give further insights into their properties.  Table~\ref{tab:comparisonofrules} summarizes the properties satisfied by our mechanisms compared to other mechanisms in this setting, and highlights how our new mechanisms have striking advantages over existing mechanisms in the literature. 
In contrast to two prominent mechanisms for this setting, our mechanisms satisfy Pareto-efficiency for both sides, strategyproofness for both sides, and non-bossiness.\footnote{Some examples and omitted proofs are in the appendix.} 

\begin{table}[h!]
	\centering
	\scalebox{0.9}{%
		\begin{tabular}{lccccc}
			\toprule
			
			&DA& REV&\textbf{SAFE/Rank-Max} \\
			Properties&&& \\
			\midrule
			Individual rationality&\checkmark&\checkmark&\checkmark\\
			Fairness&\checkmark &\checkmark&\checkmark\\
			Institution Pareto-efficiency&--&--&\checkmark\\
			Agent Pareto-efficiency&--&\checkmark&\checkmark\\
			Both-sided Pareto-efficiency&--&--&\checkmark\\
			Weak agent strategyproofness &\checkmark&\checkmark&\checkmark\\
			Agent strategyproofness &\checkmark&--&\checkmark\\
			Institution strategyproofness &--&--&\checkmark\\
			Both-sided strategyproofness&--&--&\checkmark\\
			Non-bossiness&--& -- &\checkmark\\
			\bottomrule
			
		\end{tabular}
	}
	\caption{Axioms satisfied by various mechanisms. The SAFE/Rank-Max mechanisms are introduced in this paper. DA is the classical Deferred Acceptance rule applied in our context with tie-breakers. 
		REV was introduced by \citet{AzBr24a}.
	}
	\label{tab:comparisonofrules}
\end{table}

\section{Related Work}

Two-sided matching under preferences has a long history~\citep{Manl13a,RoSo90a}. Classical results focus on strict preferences on both sides~\citep{GaSh62a,Roth08a}.  When both sides have strict preferences, stability and strategyproofness from both sides is impossible~\citep{Roth82a}, and  Pareto efficiency and stability are incompatible~\citep{AbSo03a}.  The classical deferred acceptable algorithm can be applied to our problem as follows: break the ties in the dichotomous preferences of the agents to convert them into strict preferences and then run agent proposing deferred acceptance algorithm. However, the approach does not necessarily give a maximum size or Pareto optimal matching. It is also only strategyproof for one of the sides whereas we establish strategyproofness for both sides. Similarly, the Top Trading Cycles algorithm is another algorithm for matching under preferences but if applied to our context, it does not satisfy the maximum size property or fairness. 

There is also work on two-sided matching where both sides have dichotomous preferences. For example, \citet{BoMo04a} presented several results on randomized matching under dichotomous preferences. The results do not apply to our setting where one side has strict preferences. For example when both sides have dichotomous preferences, fairness or respect or priorities has very little bite. We also focus on deterministic matchings and are able to achieve several axiomatic properties without resorting to randomisation. \citet{Aziz17b} proposed rules for exchange problems when agents have dichotomous preferences. The model does not consider priorities of objects. 

One of the key properties that we focus on is computing a feasible \textit{maximum size} matching. The assignment maximization problem is not only relevant for daycares, but also for schools~\citep{APR05a,basteck2015matching}, and for the allocation of any goods which are in shortage or need to be rationed, such as public housing, vaccines and organs. \citet{RSU05a} study the kidney exchange problem to find a maximal and strategyproof mechanism. \citet{ergin2017dual}, \citet{andersson2020pairwise}, and \citet{ergin2020efficient} aim to maximize the number of patients receiving transplants in organ exchange including kidneys, lungs, and liver. Achieving a maximum and efficient matching between refugees and landlords is also an increasingly important problem in market design \citep{andersson2020assigning, delacretaz2019matching}. Another application in which assignment maximization of objects to agents is of significant concern is the house allocation problem \citep{Aziz17b, krysta2014size, abraham2005pareto}. One particular problem for which assignment maximization is important is healthcare rationing where we want to utilize the maximum number of healthcare resources. We discuss the connections below.


\citet{PSU+20a,PSU+23a} consider a healthcare rationing in which agents have types and they are matched to categories pertaining to particular types. An agent can be matched to a category if it satisfies some type that the category is dedicated to. The healthcare rationing problem can be abstracted to our model by ignoring the types and simply assuming that an agent and category find each other acceptable if they can be matched to each other in the healthcare problem. \citet{PSU+20a} focus on homogenous priorities whereas we allow institutions to have heterogeneous preferences. A standard approach for the problem is to treat reserves from categories in a sequential manner~\citep{KoSo16a,DPS20a,AyBo20a,AyTu20b}. These approaches violate axioms pertaining to neutrality or fairness towards categories. The myopic picks can also lead to outcomes that do not satisfy the maximum size property.

\citet{AzBr21a,AzBr21b,AzBr24a} consider a healthcare rationing problem with heterogeneous priorities. 
\citet{AzBr21a} showed that maximum size individually rational matching that satisfies respect for priorities can be computed in polynomial ti     me. The result is in contrast to the fact that when both sides have weak preferences/priorities, then the problem of computing a maximum size fair matching is NP-hard. Their Reverse Rejecting (REV) rule works by considering agents in the reverse order of a baseline ordering of agents and iteratively deciding whether the agents are to be removed from consideration or not.  In followup work, \citet{BEK23a} provide an algorithmic characterization of all valid allocations, exhibiting a bijection between sets of agents who can be allocated and maximum-weight matchings under carefully chosen rank-based weights.

\citet{AzBr24a} prove that their algorithm is strategyproof in the following sense: no agent can express some institution as unacceptable in order to get an advantage. Their strategy space also allows for agents to lower themselves in the priority ordering of the institutions. 
Since they examine healthcare rationing setting, they do not allow agents to falsely make themselves eligible for an institution for which they are eligible. In our setup, we allow agents to express acceptable institutions as unacceptable or express as acceptable those institutions that find unacceptable.  In our setup,  we explore strategyproofness whereby no agent wants to declare an acceptable institution as unacceptable or an unacceptable institution as acceptable. 
Another aspect that is overlooked in most of the previous work is Pareto-efficiency of the institutions. We design rules that have two additional advantages over previous rules: they are strategyproof for both agents and institutions, and they are Pareto-efficient for institutions. Finally, the REV rule of Aziz and Brandl violates non-bossiness. Aziz and Brandl posed the question whether there exists a rule satisfying maximum size, fairness, agent strategyproofness and non-bossiness. We show that our rules satisfy all these properties.

\section{Model with Dichotomous Agent Preferences}

Our model is a two-sided matching model between agents and institutions. Agents have dichotomous preferences over institutions, while institutions strictly rank all agents according to their preferences. Our setting has the following components. 

\begin{itemize}
	
	\item Set of $n$ agents: $N$
	
	\item Set of $m$ institutions: $D$
	
	\item Dichotomous agent preferences: for all $i\in N$, $A_i \subseteq D$ is the set of acceptable institutions for $i$. The set of acceptability reports is represented by the acceptability profile $A=(A_1,\ldots, A_n)$.
	
	\item Strict institution preferences: for all $d \in D$, $\succ_d$ denotes the preference list of institution $d$ over agents $N$. Institutions are assumed to find all agents acceptable.  The set of strict preference reports by institutions is represented by the preference profile $\succ=(\succ_{d_1},\ldots, \succ_{d_m})$.
\end{itemize}

We assume that each institution has unit capacity. This is to simplify the presentation. 
The results also hold if each institution may have higher capacity in which case each institution can be viewed as being divided into smaller institutions with unit capacity each. This is because the assumption that each slot at an institution is represented by a separate institution leads to a larger preference domain than the general multi-unit case, since with unit capacity each institution can be reported acceptable/unacceptable by each agent independently of other institutions, while for the multi-unit case all slots at the same institution are either acceptable or not. 

All the information is captured in a problem instance ($N, D, A, \succ)$. If we assume that $N$ and $D$ are fixed, then a problem is given by a \textit{profile} $(A,\succ)$, consisting of an acceptability profile for the agents and a preference profile for the institutions. 
We are interested in matching agents to institutions. Each agent is either matched to some institution or remains unmatched. Each institution is either matched to some agent or remains unmatched. A matching specifies which agent is matched to which institution, and which agents and institutions remain unmatched. Each acceptability profile $A$ has a simple graph representation. A problem gives an underlying \textit{acceptability graph} $G=(N \cup D, E)$ where for all pairs $i \in N$ and $d\in D$,  $\{i,d\}\in E$ if and only if  $d\in A_i$. A matching is \textit{individually rational} if it is a matching in graph $G$. We want to find a matching that is individually rational and maximum-size, which can be determined from the acceptability graph. 
In addition to individual rationality and maximum size, we also want the matching to be Pareto-efficient for the institutions. In order to determine Pareto-efficiency for the institutions, we need more information than just $G$, we also need to know the institutions' preference profile $\succ$. 

Given a profile $(A, \succ)$, we call the preference list of institution $d \in D$ that consists of only the agents who find institution $d$ acceptable, and which follows the priority ordering $\succ_d$,  the \textit{acceptance list} of institution $d$. Formally, the \textbf{acceptance list} of institution $d \in D$ is the ordered list $\succ_d^A$ of the agents in $\{ {i \in N: \, d \in A_i} \}$ for which, for all $i,j \in N$ such that $d \in A_i \cap A_j$, we have $i \succ_d^A j$ if and only if $i \succ_d j$. Note that  $\succ_d^A$ is a function of $A$ in addition to $\succ_d$, since only agents that report an institution acceptable appear in the institution's acceptance list. The set of acceptance lists for all institutions is the \textit{acceptance list profile} $\succ^A$.  	

We illustrate the matching problem by the next example.

\begin{example} \label{Ex-Intro}
	
	Suppose there are four agents $1,2,3,4$ and two institutions $d_1$ and $d_2$. So $N=\{1,2,3,4\}$ and $D=\{d_1,d_2\}$.
	Each agent specifies its set of acceptable institutions:
	\[A_1=\{d_1,d_2\}, \; A_2=\{d_2\}, \; A_3=\{d_1\}, \; A_4=\{d_1,d_2\}.\]
	
	The two institutions have the following preferences in decreasing order of preferences from left to right: 
	\begin{align*}
		\succ_{d_1}: &\quad 2,4,1,3\\
		\succ_{d_2}: & \quad 1,3,4,2
	\end{align*}
	
	The acceptance lists are:
	\begin{align*}
		\succ^A_{d_1}: &\quad 4,1,3\\
		\succ^A_{d_2}: & \quad 1,4,2
	\end{align*}
	
	One desirable matching is $\{ (d_1,4), (d_2,1)\}$, which matches agent $4$ to $d_1$ and agent $1$ to $d_2$. Agents 2 and 3 are unmatched. 
\end{example}

\section{Axioms}

A matching is \textit{individually rational} if for each $i\in N$, $d\in D$ that are matched to each other, it must be that $d\in A_i$.  A matching is \textit{maximum size} if there is no other matching that matches more agents to institutions. A matching $\mu$ is \textit{Pareto-efficient for agents} (or $N$\textit{-efficient}, for short) if there is no matching $\mu'$ in which no agent is worse off and at least one agent strictly prefers $\mu'$ to $\mu$. In other words, $N$-efficiency requires that the matching is not Pareto-dominated for agents. Under dichotomous preferences, it is well-known that a matching is $N$-efficient if and only if it is a maximum matching. Moreover, note that $N$-efficiency implies individual rationality, since a matching that includes any agent assigned to an unacceptable institution is Pareto-dominated by the matching which leaves such agents unmatched but otherwise makes the same matchings.

	A matching is \textit{fair} (or, equivalently, respects the preferences of the institutions) if the following scenario does not arise: there exists $i,j\in N$ and $d\in D$ such that $i$ is unmatched, $d\in A_i$, and $i\succ_d j$. 
	Note that when preferences are strict the notion of fairness corresponds to the standard stability notion for two-sided matching models, and also to \textit{no justified-envy} for one-sided matching models. It means that there is no priority violation for any agent in the matching.

	A matching $\mu$ is \textit{Pareto-efficient for institutions} (or $D$\textit{-efficient}, for short) if there is no matching $\mu'$ in which no institution is worse off and at least one institution  strictly prefers $\mu'$ to $\mu$. Note that a a matching that is Pareto-effient for institutions is fair, but the converse is not necessarily true. To illustrate the latter, consider matching $\{\{d_1,1\},\{d_2,4\}\}$ in Example~\ref{Ex-Intro}. This matching is fair but not $D$-efficient.

	A matching $\mu$ is \textit{Pareto-efficient for both agents and institutions} (or \textit{efficient,} for short) if there is no matching $\mu'$ in which no agent or institution is worse off and at least one agent or institution strictly prefers $\mu'$ to $\mu$. While in general Pareto-efficiency for each set of a partition of the agents separately does not imply Pareto-efficiency for the entire set of agents, in our model the conjunction of $N$-efficiency and $D$-efficiency implies efficiency. This follows because $N$-efficiency implies individual rationality, and individually rational maximum matching cannot be Pareto-dominated by matching a different set of agents, and thus a $D$-efficient individually rational maximum matching is efficient.    
	
	A \textit{mechanism} $f$ is a function which matches agents to institutions at each profile $(A,\succ)$ such that no more than one agent is matched to an institution and each agent is matched to at most one institution. All of the above definitions of axioms for individual matchings are extended to mechanisms in the usual manner.  For example, a matching {mechanism}  is \emph{fair} if it assigns a fair matching to each profile $(A, \succ)$. 
	
	We also aim for \textit{strategyproofness}: no agent has an incentive to change its preference to get matched; and no institution has an incentive to change its preference to get a more preferred matched agent.
	A mechanism is \emph{strategyproof for the agents} (or $N$-strategyproof, for short) if no agent can misreport its preferences to obtain a strictly better outcome.  A mechanism is  \textit{strategyproof for the institutions} (or $D$-strategyproof, for short) if no institution can misreport its preferences to obtain a strictly better outcome. 
	
	Observe that any mechanism that yields a maximum size matching would not be $D$-strategyproof if institutions were allowed to report some agents unacceptable. To see this, consider a case in which two institutions have identical preferences over acceptable agents 1 and 2: $1\succ_{d_1} 2$ and $1\succ_{d_2} 2$. The institution that is assigned 2 can ensure it gets 1 by reporting 2 as unacceptable. Hence we are able to ensure that a mechanism is $D$-strategyproof only because institutions need to accept all agents.  This is a standard and reasonable assumption for institutions, since daycares and schools typically cannot exclude any children from attending them. Similarly, if the institutions represent healthcare equipment or vaccines, patients and customers cannot be excluded from access based on the preferences of healthcare providers. Therefore, we require that all reported preference orderings $\succ_d$ by institution $d$ rank all agents in $N$.

	\subsection*{Existence of a Maximum and Fair Matching}

	In matching models with strict preferences a fair and maximum matching does not necessarily exist. It is well known that in a one-to-one or many-to-one two-sided matching model all stable matchings have the same cardinality, which follows from the Rural Hospital Theorem  \citep{mcvitie1970stable, roth1985college}. Moreover, one can easily construct examples where stable matching are not maximum matchings. For example, if higher-priority agents have more acceptable schools than lower-priority agents, and both priorities and preferences are homogeneous, it is possible that higher priority agents are assigned to all the institutions that are acceptable to lower-priority agents, based on the strict preferences of higher-priority agents, leaving lower-priority agents unassigned; this is not a maximum matching.  
	Since stability corresponds to fairness in one-sided matching models, this means that fairness cannot be reconciled with maximum matchings when preferences are strict. By contrast, we  point out that in our model with dichotomous preferences a maximum and fair matching always exists. 
	
	\begin{observation} \label{Prop-MaxFairExist}
		There exists a maximum and fair matching at every profile. 
	\end{observation}
	
	\begin{proof}
		Fix a profile $(A, \succ) \in {\mathcal A} \times \Pi$. Fix a maximum matching $\mu^0$ at $(A , \succ)$. If $\mu^0$ is not fair then there exists an agent-institution pair  $(c_1,d_1)$ such that $d_1 \in A_{c_1},$ $\mu^0_{c_1} = 0$ and, given that $\mu^0$ is a maximum matching, there exists an agent $c'_1$ such that $\mu^0_{d_1} = c'_1$ and $c_1 \succ_{d_1} c'_1$. Assign $c_1$ to $d_1$, and let the other assignments be the same as in $\mu^0$. Call this matching $\mu^1$. If $\mu^1$ is not fair then there is an agent-institution pair $(c_2,d_2)$ such that $d_2 \in A_{c_2},$ and, given that $\mu^1$ is a maximum matching, there exists an agent $c'_2$ such that $\mu^1_{d_2} = c'_2$ and $c_2\succ_{d_2} c'_2$. Assign agent $c_2$ to $d_2$, and let the other assignments be the same as in $\mu^1$. Call this matching $\mu^2$. Keep repeating the same argument and apply a similar modification to the matching iteratively. Observe that an agent-institution pair cannot be repeated in this sequence, since in each step the priority of the new agent assigned to an institution is higher than the priority of the previous agent who was assigned to this institution. Moreover, the number of such improvement steps is finite, given that there is a finite number of agents and a finite number of institutions. When we can no longer find such an agent-institution pair for some matching $\mu^k$ ($k \geq 0)$, the matching $\mu^k$ is fair by definition. Moreover, since each institution that had an agent assigned to it in $\mu^0$ still has an agent assigned to it in $\mu^k$, implying that $|\mu^0| = |\mu^k|$, and given that $\mu^0$ is a maximum matching, it follows that $\mu^k$ is also a maximum matching. Thus, we have shown that there exists a maximum and fair matching for an arbitrary profile $(A, \succ) \in {\mathcal A} \times \Pi$. 
	\end{proof}
	
	Given that a maximum matching is Pareto-efficient for agents, \ref{Prop-MaxFairExist} implies that Pareto-efficiency for agents can be reconciled with fairness, another result which does not hold when preferences are strict. This positive result is possible due to the indifferences in dichotomous preferences. Proposition~\ref{Prop-MaxFairExist} also follows from Theorem~\ref{ThmMax} and Theorem~\ref{ThmFair} together (see Section~\ref{sectionSAFEProp}) and is implied by the results of \cite{AzBr21a,AzBr21b} as well.  We included an intuitive direct proof here to verify this important observation directly.

	\section{SAFE Mechanisms}
	
	We introduce a class of mechanisms called SAFE that is based on the idea of `safe blocks'. 
	
	\subsection{The Acceptability Graph and Safe Blocks}

	Given a profile $(A,\succ)$, we define the following concepts for sets of institutions for any $k$ such that $1 \leq k \leq m$. In the definitions below, we exclude institutions that are not acceptable to any agent at $(A,\succ)$. We refer to such institutions as \textit{null-institutions}, since they have an empty acceptance list.  
	
	\medskip
	
	\noindent \textbf{Equal-acceptable set of institutions:}
	
	\smallskip
	
	\noindent A set of $k$ institutions that have exactly $k$ agents on their acceptance lists collectively.
	
	\smallskip
	
	\noindent \textbf{Under-acceptable set of institutions:}
	
	\smallskip
	
	\noindent A set of $k$ institutions that have less than $k$ agents on their acceptance lists collectively.
	
	\noindent \textbf{Over-acceptable set of institutions:}
	
	\smallskip
	
	\noindent A set of $k$ institutions that have more than $k$ agents on their acceptance lists collectively.
	
	\medskip

\medskip

\noindent \textbf{Safe} \bm{$k$}\textbf{-block:}

\smallskip

\noindent Given a profile $(A,\succ)$ and a set of institutions $\bar{D} \subseteq D$, a \textbf{safe} \bm{$k$}\textbf{-block},  or simply a \textbf{safe block,} in $\bar{D}$ is an equal-acceptable set of $k$ institutions $\tilde{D} \subseteq \bar{D}$, and there is no proper subset of $\tilde{D}$ which is also an equal-acceptable set. 

There may not exist any equal-acceptable set of institutions in $\bar{D}$ at a specific profile, which would imply that there is no safe block in  $\bar{D}$. On the other hand, there may be multiple safe blocks at a profile in $\bar{D}$, and safe blocks may even overlap. By definition, a safe block cannot have a proper subset which is also a safe block. We will show below that it is always feasible to match the $k$ agents to the $k$ institutions in a safe $k$-block in a one-to-one manner (see Lemma~\ref{L2}). However, it is easy to observe from the examples below that when safe blocks overlap it is not necessarily feasible to match some agent to each institution that is in a safe block, but all agents that are on the acceptance list of some institution in a safe block are matched to an institution in any maximum matching. Moreover, observe that safe blocks only depend on the underlying acceptability graph and are independent of the strict preferences of institutions, that is, they only depend on $A$ and are independent of $\succ$. 

\medskip

\begin{example} \rm{\textbf{Safe blocks}} \label{Ex-SafeBlock}
\end{example}
\noindent
Consider the following problems, given by their acceptance lists: 

\medskip

\begin{table}[htp]
	\centering
	\begin{tabular}{l l l l l l l}
		\textbf{Problem 1:} & & \textbf{Problem 2:} & &\textbf{ Problem 3:} & &  \textbf{Problem 4:} \\
		$d_1: 1, 3$ & \; \ & \, $d_1: 2, 1$ & \;  & \; $d_1: 2, 1, 3$ & \; & \;  $d_1: 1, 2$ \\
		$d_2: 3, 2$  & \; & \, $d_2: 1, 2$ & \; & \, $d_2: 1, 2$  & \;  & \;    $d_2: 2,1,3$  \\
		$d_3: 1$  & \;  & \, $d_3: 1, 2, 3$ & \; & \; $d_3: 1, 2, 3$ & \; & \;   $d_3: 3, 1, 2$ \\
		& \;  & \, & \; & \;  & \; & \;  $d_4: 3, 1$ \\
	\end{tabular}
\end{table} 

\begin{description}[nosep]
	
	\item[Problem 1:] Institution $\{d_3\}$ is a safe 1-block. There are no other safe blocks. In particular, although $\{d_1, d_2, d_3\}$ and $\{d_1, d_3\}$ are both equal-acceptable, these are not safe blocks because $\{d_3\}$ is also equal-acceptable.  
	
	\item[Problem 2:] The set  $\{d_1, d_2\}$ is a safe 2-block. 
	
	\item[Problem 3:] The three institutions $\{d_1, d_2, d_3\}$ together constitute a safe 3-block, since each proper subset is over-acceptable.

	\item[Problem 4:] There are four safe blocks: any three institutions constitute a safe 3-block. Observe that all three agents can be assigned to an institution, but one of the institutions will be unmatched. 
	\hfill $\diamond$
	
\end{description}

\medskip
\begin{example} \rm{\textbf{No safe block}}
\end{example}
\noindent
There are eight agents $(n=8)$ and six institutions $(m=6)$ with the following acceptance lists: 

\medskip

$ \; \; \; d_1: 1, 2, 3$

$\; \; \; d_2: 3, 2, 4$

$\; \;\ \; d_3: 1, 3$

$\; \;\ \; d_4: 1, 3, 4, 5$

$\; \;\ \; d_5: 6, 7, 4$

$\; \;\ \; d_6: 6, 8$

\bigskip

\noindent There is no set of institutions that is equal-acceptable in this problem, so there is no safe block. Each subset of the institutions is over-acceptable. 
\hfill $\diamond$

\subsection{Definition of the SAFE Mechanisms}

Let $\pi$ be a fixed ordering (i.e., a permutation) of the set of institutions $D$. Given a profile $(A,\succ)$, a \bm{$\pi$}\textbf{-sequential matching} at $(A,\succ)$ is the matching reached by iteratively assigning each institution in the order of $\pi$ the highest-priority agent on its acceptance list who is still unassigned, if there is any. We will also say that, given a subset of the institutions $\bar{D} \subseteq D$, the \bm{$\pi$}\textbf{-sequential \bm{$\bar{D}$}-matching} at $(A,\succ)$ is the matching restricted to $\bar{D}$ which is reached by iteratively assigning each institution in $\tilde{D}$ in the order of $\pi$ the highest-priority agent on its acceptance list who is still unassigned, if there is any. A matching is a \textbf{sequential matching} at $(A,\succ)$ if there exists a permutation $\pi$ of $D$ such that it is the $\pi$-sequential matching at $(A,\succ)$. A \textbf{sequential mechanism} is a mechanism which assigns a sequential  matching at $(A,\succ)$ to each profile $(A, \succ) \in {\mathcal A} \times \Pi $. Note that the permutation $\pi$ may vary with the profile $(A, \succ)$, and thus sequential mechanisms are not simply serial dictatorships \citep{satterthwaite1981strategy} that are constrained by the acceptability graph. 

A \textbf{SAFE (Sequential Allocation for Fairness and Efficiency) mechanism} is a sequential mechanism for which a fixed baseline permutation $\bar{\pi}$ over $D$ is used at each profile $(A,\succ)$ to determine the profile-dependent permutation $\pi (A, \succ) $ which leads to the $\pi$-sequential matching. The permutation $\pi$ is a specific ordering of institutions which depends on the acceptance lists of institutions at $(A, \succ)$ and on the baseline ranking $\bar{\pi}$. In each step of the iterative procedure there is either at least one safe block or no safe block, based on the remaining set of institutions with their acceptance lists containing only the remaining agents, and this together with the tie-breaking permutation $\bar{\pi}$ determines iteratively the next institution in $\pi (A,\succ)$ in each step of the procedure, which yields the $\pi (A,\succ)$-sequential matching. 
Thus, each SAFE mechanism is specified by the fixed baseline ranking $\bar{\pi}$, and hence the class of SAFE mechanisms is given by the $m!$ permutations of $D$. We will denote the SAFE mechanism with tie-breaking permutation $\bar{\pi}$ by $\varphi^{\mbox{\footnotesize SAFE}(\bar{\pi})}$. We define the mechanism formally as  Mechanism~\ref{mech:SAFE}.

\bigskip

%

\begin{algorithm}[h!]
	\KwIn{$(N,D,A,\succ)$ and baseline permutation $\bar{\pi}$ over $D$}
	\KwOut{A matching $M$}
	Initialize $G$ as the acceptability graph for $(N,D,A,\succ)$.\\
	Initialize matching $M$ to empty set.\\
	\While{$G$ is not the empty graph}
	{\If{there is a safe block in $G$}{
			find the highest ranked (according to permutation $\bar{\pi}$) institution $d$ \\in graph $G$ that is in a safe block and match the most preferred agent $i$ in $G$ (according to $\succ_d$) to d: add $\{i,d\}$ to $M$. Remove $i$ and $d$ and all adjacent edges from $G$\\}
		\Else{take the highest ranked (according to permutation $\bar{\pi}$) institution $d$ \\in graph $G$ and match to the most preferred agent in $G$ (according to $\succ_d$) to d:  add $\{i,d\}$ to $M$. Remove $i$ and $d$ and all adjacent edges from $G$}}
	\Return $M$
	
	\caption{\textbf{SAFE}}\label{mech:SAFE}
\end{algorithm}








\subsection{SAFE Mechanisms: An Illustrative Example}

The examples below demonstrate how SAFE mechanisms work. The updated acceptance lists are displayed for each remaining institution in each step after the initial step.  

\begin{example} Illustration of the SAFE mechanism. \label{Ex-SAFE}
\end{example}

\noindent Let $\bar{\pi}= (d_1, d_2, d_3, d_4)$, and consider the following acceptance list profile.

\medskip

\begin{center} 
	
	\begin{tabular}{l}
		
		$ \; \; \; d_1: 1, 2, 3$\\
		$\; \; \; d_2: 1, 2, 3$\\
		
		$\; \;\ \; d_3: 1$\\
		
		$\; \;\ \; d_4: 2$\\
		
	\end{tabular}
	
\end{center}

\medskip

\noindent There are two safe blocks in step~1, consisting of institution $d_3$ and $d_4$ respectively, and agent 1 is assigned to $d_3$ based on $\bar{\pi}$. After removing agent $1$ and institution $d_3$ from the problem, the updated acceptance list profile yields a new safe block: $\{d_1, d_2\}$. Thus, agent 2 is assigned to institution $d_1$ in step~2, and thus agent 3 is assigned to institution $d_2$ in step~3. The steps are illustrated below.

\medskip

\begin{table}[htp]
	\begin{tabular}{l l l l l l l l l}
		\textbf{Step 1:} & & \textbf{Step 2:} & &\textbf{ Step 3:} \\
		$d_1: 1, 2, 3$ & \;  & \; $d_1: \fbox{$2$}, 3$ & \;  & \; $ $ \\
		$d_2: 1, 2, 3$  & \; & \; $d_2: 2, 3$ & \; & \; $d_2:$ \fbox{$3$}\\
		$d_3:$ \fbox{$1$} & & & & \\
		$d_4:$ 2 &\;  & \; $d_4:$ 2 &\; & \\ 
	\end{tabular}
\end{table} 

\medskip

The final matching is $\{ (d_1,2), (d_2,3), (d_3,1), (d_4, 0)$.

\medskip

\noindent Note that in this problem the maximum matching is selected with the aid of safe blocks in each step, while the baseline ranking $\bar{\pi}$ plays a role in selecting the first agent in a safe block in each step. The result is the $\pi$-sequential matching with $\pi = (d_3, d_1, d_2)$.
\hfill $\diamond$

\medskip

\section{Rank-Maximal Mechanisms: An Equivalent Family of Mechanisms}

We now introduce another family of mechanisms, which we call Rank-Maximal, and show that it is equivalent to the family of SAFE Mechanisms. Rank-Maximal is a family of rules that are parametrized by a baseline permutation of institutions, just like SAFE mechanisms.   Based on a baseline permutation $\pi$ over institutions, a set $S\subset D$ is\textit{ lexicographically better} than another set $T\subset D$ if for the earliest (according to $\pi$) institution where $S$ and $T$ differ, it is the case that the corresponding institution in $S$ comes earlier according to $\pi$. The Rank-Maximal algorithm focusses on matching a lexi-optimal set of institutions that can call be matched. We say that such a set is\textit{ lexi-optimal}. The mechanism follows the ordering of the baseline permutation $\pi$ and iteratively matches each institution to the most preferred agent such that the resulting matching can be extended to a matching that matches all institutions in a lexi-optimal set of institutions. In the definition of the rule we first determine the set of institutions that will be matched, based on the acceptability graph of the problem. Then, given the corresponding restricted acceptability graph, we determine the assignments step-by-step such that each institution is assigned the most preferred agent that allows for matching all the institutions in the restricted graph. We will denote the Rank-Maximal mechanism with baseline permutation $\bar{\pi}$ by $\varphi^{\mbox{\footnotesize RankMax}(\bar{\pi})}$. The mechanism is formally specified as Mechanism~\ref{mech:RM}.

%
%
%

\begin{algorithm}[h!]  
	\KwIn{$(N,D,A,\succ)$ and permutation $\pi$ over $D$}
	\KwOut{Matching $M$}
	Set $W$ to empty set and let $G$ be the acceptability graph for $(N,D,A,\succ)$. \\
	\While{$\exists d\in D\setminus W$ such that $W\cup \{d\}$ can be matched in some matching in $G$}{
		Find the highest priority (according to $\pi$) such $d$\;
		Add to $W$ the highest priority (according to $\pi$) institution $d$ in $W$ such that  all elements of $W\cup \{d\}$ can be matched in some matching in $G$}
	Initialize $G$ as the acceptability graph for $(N,W,A,\succ)$ restricted to $W\subseteq N$\;
	Initialize matching $M$ to empty\;
	\While{$|M|<|W|$}{
		For the highest ranked institution $c\in D$ (according to $\pi$), find the highest priority agent (according to $\succ_c$) such that if we remove $c$ and $i$ from $G$, then the modified graph $G$ admits a matching that matches every institution\;
		Remove $i$ and $c$ from $G$\;
		Permanently match $i$ to $c$: add $\{i,c\}$ to $M$\;
	}
	\Return $M$.
\caption{\textbf{Rank-Maximal}}\label{mech:RM}
\end{algorithm}

For a permutation $\pi$ over institutions, consider the \textit{ranked acceptability graph} $(N \cup D, E, \ell)$ where the edge set represents the individual rationality relations 
and each edge is assigned a rank as follows. Suppose the edge set $E$ is partitioned into $r$ disjoint sets, i.e., $E = E_1\cup E_2 \cup \cdots \cup E_{|D|}$ where $E_i$ is the set of edged adjacent to institution $d=\pi(i)$ which are given a rank $i$.  
The \emph{signature} $\rho(M) = \langle x_1, x_2, ..., x_{r} \rangle$ of a matching $M$ in $G$ is a tuple of integers where each element $x_i$ represents the number of edges of rank $i$ in $M$. 

For a ranked bipartite graph, we compare the signatures of matchings in a lexicographical manner. A matching $M'$ with $\rho(M') = \langle x_1, \cdots, x_{r} \rangle$ is 
\textbf{\textit{strictly better}} 
than another matching $M''$ with $\rho(M'') = \langle y_1, \cdots, y_{r} \rangle$, if there exists an index $1 \leq k \leq r$ s.t. for $1 \leq i < k$, $x_i = y_i$ and $x_k > y_k$. 
A matching $M'$ is \textbf{\textit{weakly better}}  than another matching $M''$ if $M''$ is not strictly better  than $M'$. 
Let $M' \succ_{lex} M''$ denote that $M'$ is strictly better than $M''$ and let $M' \succsim_{lex} M''$ denote that $M'$ is weakly better than $M''$.

A matching $M$ in a ranked bipartite graph $G$ is \textbf{\textit{rank-maximal}} if $M$ is weakly better than any other matching $M'$ in $G$. 
A rank maximal matching can be computed in polynomial time~\citep{Manl13a}. 

\begin{lemma} \label{lemma:RM1}
A set of institutions matched in a rank maximal matching of a ranked acceptability graph is the lexi-optimal set of institutions. 
\end{lemma}
\begin{proof}
By definition, a rank maximal matching's matched institutions are lex-optimal. 
\end{proof}

Although a rank maximal matching of a bipartite graph need not be a maximum size matching in generalm for our particular problem that has more structure,  a rank maximal matching gives a maximum size matching.

\begin{lemma}\label{lemma:RM2}
For our problem, a rank maximal matching gives a maximum size matching. Equivalently, if a matching matches the lexi-optimal set of institutions, then the matching is maximum size.
\end{lemma}
\begin{proof}
Suppose the rank maximal matching of graph $(N\cup D, E, \ell)$ is not of maximum size.  By  Berge's lemma, it admits an augmenting path $p$. Now suppose we switch the matching $M$ to $M'$ by removing the matched edges in the augmenting path $p$ and selecting the complement of the edges in $p$. Note that those vertices that were matched in $p$ under $M$ continue to be matched under $M'$ and there is at least one additional edge. Therefore $M'$ is better than $M$ in terms of rank maximality which is a contradiction. After this step, institutions may exchange agents but the size of the matching does not change. 
\end{proof}

Next, we illustrate how a Rank-Maximal mechanism works.

\begin{example}[Rank-Maximal Mechanism]

We illustrate how RM works.

\begin{center} 

\begin{tabular}{l}
	
	$ \; \; \; d_1: 1, 2, 3$\\
	$\; \; \; d_2: 1, 2, 3$\\
	
	$\; \;\ \; d_3: 1$\\
	
	$\; \;\ \; d_4: 2$\\
	
\end{tabular}

\end{center}
The lexi-optimal set of institutions $W$ that will be matched, as computed by Rank-Maximal in the initial step of the algorithm, is $\{d_1,d_2,d_3\}$.
\[W=\{d_1,d_2,d_3\}.\] 
For $d_1$, it cannot be matched to $1$ because it does, then not all members of $\{d_1,d_2,d_3\}$ can be matched. Hence $d_1$ is matched to $2$. 

\begin{center} 
\begin{tabular}{l}
	$d_1:$ 1, \fbox{$2$}, 3 \\
	$d_2:$ 1, 2, {$3$} \\
	$d_3: $ {$1$}\\
	$d_4:$ 2\\
\end{tabular}
\end{center}

Institution $d_2$, cannot be matched to $1$ while $d_1$ is matched to $2$ as then $d_3$ cannot be matched. Also, $d_2$, cannot be matched to $2$ as $d_1$ is already matched to $2$. Hence $d_2$ is matched to $3$. 

\begin{center} 
\begin{tabular}{l}
	$d_1:$ 1, \fbox{$2$}, 3 \\
	$d_2:$ 1, 2, \fbox{$3$} \\
	$d_3: $ {$1$}\\
	$d_4:$ 2\\
\end{tabular}
\end{center}

Finally, $d_3$ is matched to $1$. 

\begin{center} 

\begin{tabular}{l}
	
	\textbf{Final matching:}\\
	
	$d_1:$ 1, \fbox{$2$}, 3 \\
	$d_2:$ 1, 2, \fbox{$3$} \\
	$d_3: $ \fbox{$1$}\\
	$d_4:$ 2\\
	
\end{tabular}

\end{center}
\end{example}

Note that in this example both the SAFE and Rank-Maximal mechanisms give the same outcome. This is not a coincidence.
We prove next that the two families of mechanisms that we introduced, SAFE and Rank-Maximal, are equivalent.

\begin{proposition}
For all permutations $\bar{\pi}$ over $D$, mechanisms $\varphi^{\mbox{\footnotesize SAFE}\, (\bar{\pi})}$ and $\varphi^{\mbox{\footnotesize RankMax}\,(\bar{\pi})}$ are outcome equivalent. 
\end{proposition}

\begin{proof}
Fix a permutation $\bar{\pi}$ over $D$ and a profile $(A,\succ)$. We will show that
\[\varphi^{\mbox{\footnotesize SAFE}\, (\bar{\pi})}(A,\succ) = \varphi^{\mbox{\footnotesize RankMax}\,(\bar{\pi})}(A,\succ).\] 

\noindent \textbf{Step 1:} Suppose for a contradiction that the set of institutions that are matched by $\varphi^{\mbox{\footnotesize SAFE}\, (\bar{\pi})}$ are not the same as those of $\varphi^{\mbox{\footnotesize RankMax}\,(\bar{\pi})}(A,\succ)$. Let $\bar{\pi}(t)$ (the $t$-th-ranked institution according to $\bar{\pi}$) be the first institution according to $\bar{\pi}$ with such a discrepancy. We distinguish between the following two cases. 

\begin{itemize}

\item[] \textbf{Case 1:} \textit{SAFE matches institution $\bar{\pi}(t)$, while Rank-Maximal does not.} 

Since the set of institutions $W \subseteq D$ that are matched by Rank-Maximal is lexi-optimal with respect to $\bar{\pi}$ at each profile $(A,\succ)$, if $\varphi^{\mbox{\footnotesize SAFE}\,(\bar{\pi})}$ matches $\bar{\pi}(t)$ and Rank-Maximal does not, then there is a discrepancy for a previous institution according to $\bar{\pi}$, namely, there exists an institution $\bar{\pi}(t')$ with $t'<t$ such that Rank-Maximal matches institution  $\bar{\pi}(t')$, but SAFE does not. This contradicts our assumption on $t$. 

\item[] \textbf{Case 2:} \textit{Rank-Maximal matches institution  $\bar{\pi}(t)$, while SAFE does not.}

For ease of notation, let $d = \bar{\pi}(t)$. Since Rank-Maximal assigns $d$ at $(A,\succ)$ and SAFE mechanisms are maximum mechanisms, our assumption on $t$ implies that there exists some institution $\tilde{d}$ such that $d$ is ranked prior to $\tilde{d}$ by $\bar{\pi}$ and $\tilde{d}$ is matched by $\varphi^{\mbox{\footnotesize SAFE}\,(\bar{\pi})}$ at $(A,\succ)$. Suppose that $d$ is a null-institution (i.e., $d$'s acceptance list is empty) in all the steps when subsequent institutions according to $\bar{\pi}$ (such as $\tilde{d}$) are selected to be matched by SAFE. This is not possible, since then $\varphi^{\mbox{\footnotesize RankMax}\,(\bar{\pi})}$ could not assign $d$ at $(A,\succ)$.

Thus, institution $d$ is not a null-institution in some step $k$ of the SAFE algorithm at $(A,\succ)$ such that some institution $\tilde{d}$ is selected in step $k$ to be matched by the SAFE algorithm, where $d$ is ranked prior to $\tilde{d}$ by $\bar{\pi}$. Let $D_t \subseteq D$ denote the set of unmatched institutions in step $t$ of the SAFE algorithm that are not null-institutions in step $t$. Note that $d \in D_t$. Then each subset of $D_t$ that includes $d$ is not a safe block, otherwise $d$ would be matched by SAFE. 
Moreover, since $d$ remains unmatched by $\varphi^{\mbox{\footnotesize SAFE}\,(\bar{\pi})}$ at this profile, it follows that the above also holds in each subsequent step of the SAFE algorithm at $(A,\succ)$. However, after the last step of the SAFE algorithm the updated acceptance list of $d$ must be empty ($d$ becomes a null-institution) which implies that it is a singleton in some prior step of the algorithm. This means that $d$ is a safe block in that step, which is a contradiction. Therefore, $d$ is also matched by $\varphi^{\mbox{\footnotesize SAFE}\,(\bar{\pi})}$ at $(A,\succ)$. 

\end{itemize} 

In sum, the set of matched institutions $W \subseteq D$ is the same for  
$\varphi^{\mbox{\footnotesize SAFE}\, (\bar{\pi})}(A,\succ)$ and $\varphi^{\mbox{\footnotesize RankMax}\,(\bar{\pi})}(A,\succ)$. (Note: $W$ is lexi-optimal according to $\bar{\pi}$ by Lemma~\ref{lemma:RM1}.)

\medskip
\noindent \textbf{Step 2:}
We need to show that both $\varphi^{\mbox{\footnotesize SAFE}\,(\bar{\pi})}$ and $\varphi^{\mbox{\footnotesize RankMax}\,(\bar{\pi})}$ yield the same assignments of agents to institutions in $W$ at profile $(A,\succ)$. Since each institution is matched in $W$, safe blocks are disjoint and each agent on the acceptance list of an institution in a safe block is assigned to an institution in this safe block. Given the definitions of the SAFE and Rank-Maximal mechanisms, it follows immediately that within each safe block both $\varphi^{\mbox{\footnotesize SAFE}\,(\bar{\pi})}$ and $\varphi^{\mbox{\footnotesize RankMax}\,(\bar{\pi})}$ assign the most preferred remaining agent on the first institution's acceptance list according to the baseline ranking $\bar{\pi}$ of the institutions. The same argument applies iteratively to new safe blocks after updating. Finally, one can easily check that any remaining institutions that are not in a safe block in any step are also assigned sequentially by both algorithms in the order of $\bar{\pi}$.  This concludes the proof.      
\end{proof}

%
%
%

%
%
%

\section{Properties of SAFE/Rank-Maximal Mechanisms}

\label{sectionSAFEProp}

In this section, we establish the key axiomatic properties of SAFE/Rank-Maximal Mechanisms.



We prove four lemmas first. Lemmas~\ref{L1}-\ref{L3} show that it is feasible to assign an agent to each institution in a safe block, even after making an arbitrary first assignment. This also means that each agent that finds at least one institution acceptable in a safe block is assigned by any maximum mechanism. These lemmas, including Lemma~\ref{L4}, are used in the proof of Theorem~\ref{ThmMax} which shows that SAFE mechanisms are maximum mechanisms, while the proof of Theorem~\ref{ThmSP}, which demonstrates that SAFE mechanisms are strategyproof,  relies on Lemmas~\ref{L1} and \ref{L2}. In the arguments of the proofs of the lemmas and theorems null-institutions are excluded by default.

We also note that we will use Hall's theorem \citep{hall1987representatives} repeatedly in the proofs below. Hall's theorem states, in our terminology, that it is feasible to assign an agent to each of $k$ institutions if and only if each subset of the $k$ institutions is either equal-acceptable or over-acceptable.    We state a series of lemmas. The proofs are in the appendix.

\begin{lemma} \label{L1}
Each under-acceptable set of institutions contains a safe block.
\end{lemma}


\begin{lemma} \label{L2}
It is feasible to assign an agent to each institution within a safe block. 
\end{lemma}

\begin{proof}
By Lemma~\ref{L1}, each subset of a safe block is either equal-acceptable or over-acceptable. Then Hall's theorem  implies that it is feasible to assign an agent to each institution in a safe block.  
\end{proof}

It follows from Lemma~\ref{L2} that in any maximum matching each agent is matched who has any institution in its acceptance set that is in a safe block. However, note that it does not follow that each institution in a safe block is matched in a maximum matching. 

\begin{lemma} \label{L3}
Let an acceptance list be given for each institution in $\bar{D}$, where $\bar{D} \subseteq D$. Assign an arbitrary institution from $\bar{D}$ to an arbitrary agent on this institution's acceptance list. If $\bar{D}$ is a safe block, it is 
feasible to assign each institution in $\bar{D} \setminus {d}$ an agent other than $d$ on its acceptance list.
\end{lemma}

\begin{lemma} \label{L4}
Let an acceptance list be given for each institution in $\bar{D}$, where $\bar{D} \subseteq D$. Assign an arbitrary institution from $\bar{D}$ to an arbitrary agent on this institution's acceptance list. If $\bar{D}$ has no safe block, then it is 
feasible to have a maximum matching for $\bar{D}$ which includes this initial assignment. 
\end{lemma}


\begin{theorem} \label{ThmMax}
A SAFE mechanism is a maximum mechanism.
\end{theorem}

\begin{proof}
Fix a permutation $\bar{\pi}$ of institutions in $D$, and fix a profile $(A, \succ) \in {\mathcal A} \times \Pi$.  The SAFE mechanism $\varphi^{\bar{\pi}}$ leads to a permutation $\pi$ at profile $(A, \succ)$. Without loss of generality, let $\pi = (d_1, \ldots, d_m)$. We will show that after each step of the procedure, that is, after assigning the first remaining agent, if there is any, to institution $d_t$ on the acceptance list of $d_t$ in each step $t$, for $t=1, \ldots, m-1$, it is feasible to reach a maximum matching. Since this implies that selecting any remaining agent on the acceptance list of the last institution $d_m$ in step $m$ leads to a maximum matching (or, alternatively, if $d_m$ is a null-institution in step $m$ then the matching reached before step~m is a maximum matching), this will ensure that at the end of the SAFE mechanism procedure applied to $(A, \succ)$ the assignment made iteratively in each step results overall in a maximum matching.    

Fix $t \in \{1, \ldots, m-1\}$ and assume that after step $t-1$ in the procedure all previous assignments make it feasible to reach a maximum matching at the end. Note that this assumption holds vacuously for step 0, that is, prior to beginning the procedure, when no assignments have been made yet. Consider the following two cases based on the updated institution acceptance lists after step $t-1$: a) there is a safe block, b) there is no safe block.

\medskip	

\noindent Case a):  \textit{There is a safe block based on the updated institution acceptance lists after} \textit{step}~\mbox{$t-1$}.

\smallskip
\noindent In this case a yet unassigned institution $d$ is selected from a safe block (the first institution in a safe block according to $\bar{\pi}$), and thus it follows from Lemma~\ref{L3} that all agents on the acceptance lists of institutions in a safe block can be assigned to an institution in this safe block, regardless of the first institution that is chosen from this safe block. Since, for all $k$, a safe $k$-block is an equal-acceptable set of institutions by definition, only the $k$ agents on the acceptance lists of the $k$ institutions in the safe block can be assigned to these institutions by the individual rationality of the SAFE mechanism, and thus starting the assignment with any arbitrary institution in a safe block allows for reaching a maximum matching after making this assignment.        

\medskip

\noindent Case b): \textit{There is no safe block based on the updated institution acceptance lists after step}~$t-1$.

\smallskip
\noindent In this case Lemma~\ref{L4} implies that starting the assignment with any arbitrary yet unassigned institution $d$ allows for reaching a maximum matching after making this assignment.

Given the above arguments for both cases, it follows by induction that the SAFE mechanism selects a maximum matching at profile $(A, \succ)$. Since $(A, \succ)$ is an arbitrary profile in ${\mathcal A} \times \Pi$, this proves that the SAFE mechanism $\varphi^{\bar{\pi}}$ is a maximum mechanism for any fixed permutation $\bar{\pi}$ of the institutions.	 
\end{proof}

\begin{theorem} \label{ThmFair}
A SAFE mechanism is fair.
\end{theorem}

\begin{proof}
Fix a permutation $\bar{\pi}$ of institutions in $D$, and fix a profile $(A, \succ) \in {\mathcal A} \times \Pi$. The SAFE mechanism $\varphi^{\bar{\pi}}$ leads to a permutation $\pi$ at profile $(A, \succ)$ such that $\varphi^{\bar{\pi}}(A, \succ)$ is the $\pi$-sequential matching at $(A, \succ)$. Let $\mu \equiv \varphi^{\bar{\pi}}(A, \succ)$. Suppose, by contradiction, that there is an agent-institution pair $(c, d)$ such that i) $\mu_c = 0$, ii) $d \in A_c$,  and iii) either $\mu_d = 0$ or $\mu_d = c'$ such that $c\succ_{d} c'$. Given $\mu_c = 0$ and $d \in A_c$, when we get to the step in the algorithm where institution $d$ is next in permutation $\pi$, agent $c$ is on $d$'s acceptance list since $c$ is unassigned in $\mu$ and $d$ is acceptable to $c$, and thus $\mu_d = c'$ such that $c' \succ_{d} c$. This is a contradiction. Therefore, $\mu$ is a fair matching at $(A, \succ)$. The same argument holds for $\varphi^{\bar{\pi}}(A, \succ)$ at each profile $(A, \succ) \in {\mathcal A} \times \Pi$, and thus $\varphi^{\bar{\pi}}$ is fair for an arbitrary permutation $\bar{\pi}$ of the institutions.    
\end{proof}

\begin{theorem}
A SAFE mechanism is strategyproof {for agents}. \label{ThmSP}
\end{theorem}

\begin{proof}

Let $\psi^{\bar{\pi}}$ be a SAFE mechanism where $\bar{\pi}$ is the tie-breaking permutation of the institutions. Suppose by contradiction that $\psi^{\bar{\pi}}$ is not strategyproof. Then there exist an agent $c \in C$, a profile $(A, \succ)$ and an alternative set of acceptable institutions $A'_c$ for agent $c$ such that $\psi^{\bar{\pi}}_c (A, \succ) = 0$ and $\psi^{\bar{\pi}}_c ((A'_c, A_{-c}), \succ) \in A_c$. Let $\mu \equiv \psi^{\bar{\pi}} (A, \succ)$ and $\mu' \equiv \psi^{\bar{\pi}}_c ((A'_c, A_{-c}), \succ)$. Note that $\mu_c =0$.

\medskip

\noindent \textbf{Case 1:} \textit{There exists a permutation $\pi$ of the institutions such that $\mu$ is the $\pi$-sequential matching at $(A, \succ)$ and $\mu'$ is the $\pi$-sequential matching at $((A'_c, A_{-c}), \succ)$.}

\smallskip

In this case subtracting any acceptable institutions by reporting $A'_c$ instead of $A_c$ would not have any impact, since $c$ remains unassigned at $(A,\succ)$. Moreover, adding some acceptable institutions by reporting $A'_c$ instead of $A_c$ can only result in an assignment to an unacceptable institution at $((A'_c, A_{-c}), \succ)$ (i.e., $\mu_c \notin A_c$), which is a contradiction.    

\medskip

\noindent \textbf{Case 2:} \textit{There does not exist a permutation $\pi$ of the institutions such that $\mu$ is the $\pi$-sequential matching at $(A, \succ)$ and $\mu'$ is the $\pi$-sequential matching at $((A'_c, A_{-c}), \succ)$.}

\smallskip

In this case the difference between $A_c$ and $A'_c$ changes the permutation of institutions used in the sequential matching when we compare $\mu$ to $\mu'$. This means that $c$ either destroys at least one existing safe block at $(A, \succ)$, so that it is not a safe block at $((A'_c, A_{-c}), \succ)$, or $c$ creates at least one new safe block at $((A'_c, A_{-c}), \succ)$ compared to $(A, \succ)$ by misrepresenting her true preferences over the institutions (or both).

\begin{itemize}

\item[] \textbf{Subcase 2.1} \textit{Only subtract institutions: $A'_c \subset A_c$}

\end{itemize}

Assume first that $c$ only subtracts institutions from her acceptance set, that is, $A'_c \subset A_c$. Then if an existing safe block at $(A, \succ)$ is destroyed and thus it is no longer a safe block at $((A'_c, A_{-c}), \succ)$, this implies that agent $c$ is on the acceptance list of at least one institution in this safe block at $(A, \succ)$. Then, by Lemma~\ref{L2}, $c$ is assigned to an institution at $(A, \succ)$, since  $\psi^{\bar{\pi}}$ is a maximum mechanism by Theorem~\ref{ThmMax}. Therefore, $\mu_c \neq 0$, which is a contradiction. 

This implies that  $c$ creates at least one new safe block at $((A'_c, A_{-c}), \succ)$ compared to $(A, \succ)$. Let one such new safe block at $((A'_c, A_{-c}), \succ)$ be a safe $k$-block. Then it is an over-acceptable set of institutions at $(A,\succ)$ such that $k+1$ agents have these institutions in their acceptance sets: the $k$ agents in $C \setminus\{c\}$ who have them in their acceptance sets, as given by $A_{-c}$, in addition to agent $c$. If any of these $k$ agents, say $\tilde{c}$, is assigned to an institution in a step $t$ prior to reaching the first institution among these $k$ institutions in the SAFE procedure at $(A,\succ)$, then these $k$ institutions form an equal-acceptable set of institutions in step $t+1$. Suppose by contradiction that there is an under-acceptable subset of $k' < k$ institutions of these $k$ institutions in step $t+1$. Then these $k'$ institutions form an equal-acceptable set at step $t$, when $\tilde{c}$ was still on the acceptance lists of institutions, since by Lemma~\ref{L1} it could not have been under-acceptable, given that the $k$ institutions would constitute a safe block if we deleted $c$ from the acceptance lists. This means that the set of these $k'$ institutions is either a safe block at $(A,\succ)$ or contains a safe block, which contradicts the fact that after deleting $c$ from the acceptance lists of the $k$ institutions the set of $k > k'$ institutions becomes a safe $k$-block. Therefore, each subset of the $k$ institutions in step $t+1$ is either equal-acceptable or over-acceptable, and Hall's theorem implies that it is feasible to assign an agent to each of the $k$ institutions at $(A, \succ)$. This would mean that agent $c$ was assigned to an institution at $(A, \succ$), which is a contradiction. Therefore, we conclude that none of these $k$ agents is assigned to an institution in a step $t$ prior to reaching the first institution among the $k$ institutions in the SAFE procedure at $(A,\succ)$. In this case, however, it makes no difference whether these institutions are in a safe block and whether agents are assigned to them earlier than $\bar{\pi}$ calls for (due to $c$ reporting untruthfully that some of these institutions are not in her acceptance set), and the final matching remains the same at $((A'_c, A_{-c}), \succ)$ compared to $(A,\succ)$. This is a contradiction. 

In sum, since $A'_c \subset A_c$ and $\varphi_c ((A'_c, A_{-c}),\succ) \notin A'_c$, since $\varphi_c ((A'_c, A_{-c}), \succ)$ is a feasible matching, it follows that $\varphi_c ((A'_c, A_{-c}),\succ)=0$.

\begin{itemize}

\item[] \textbf{Subcase 2.2} \textit{Only add institutions: $A_c \subset A'_c$}

\end{itemize}

Now assume that $c$ only adds institutions to her acceptance set, that is, $A_c \subset A'_c$. 
Consider the case first where an existing safe block at $(A, \succ)$ is destroyed and thus it is no longer a safe block at $((A'_c, A_{-c}), \succ)$, given that agent $c$ adds institutions to her acceptance set. Let this safe block at $(A,\succ)$ be a safe $k$-block. As shown by Lemma~\ref{L2}, $c$ is not one of the $k$ agents on the acceptance lists of the institutions in this safe block.

If any of the $k$ agents who have these $k$ institutions in their acceptance sets, as given by $A_{-c}$, say $\tilde{c}$, is assigned to an institution in a step $t$ prior to reaching the first institution in the SAFE mechanism procedure at $((A'_c, A_{-c}), \succ)$ then the $k$ institutions in the safe block constitute an equal-acceptable set of institutions in step $t+1$. Suppose by contradiction that there is an under-acceptable subset of these $k$ institutions in this subsequent step. Then this strict subset of the $k$ institutions was equal-acceptable in step $t$, when $\tilde{c}$ was still on the acceptance lists of institutions, since by Lemma~\ref{L1} it could not have been under-acceptable, given that the $k$ institutions constitute a safe block when we delete $c$ from the acceptance lists. This means that this strict subset of the $k$ institutions is either a safe block at $((A'_c, A_{-c}), \succ)$ or contains a safe block, which contradicts the fact that after deleting $c$ from the acceptance lists of the $k$ institutions they become a safe $k$-block. Therefore, each subset of of these $k$ institutions in step $t+1$ is either equal-acceptable or over-acceptable, and Hall's theorem implies that it is feasible to assign an agent to each of the $k$ institutions at $((A'_c, A_{-c}), \succ)$. This would mean that agent $c$ is assigned to an institution at $((A'_c, A_{-c}), \succ)$. If $c$ is assigned  to an institution in the safe $k$-block then $\mu'_c \notin A_c$, and if $c$ is assigned prior to reaching the safe block then $\mu'_c = \mu_c$ and thus $\mu_c \in A_c$, implying a contradiction in both cases. Therefore, we conclude that none of these $k$ agents is assigned to an institution in a step $t$ prior to reaching the first institution in the SAFE mechanism procedure at $((A'_c, A_{-c}), \succ)$. In this case, however, it makes no difference whether these institutions are in a safe block and whether agents are assigned to them earlier than $\bar{\pi}$ calls for (due to $c$ reporting untruthfully that some of these institutions are not in her acceptance set), and the final matching remains the same at $((A'_c, A_{-c}), \succ)$ compared to $(A,\succ)$. This is a contradiction.

Finally, if $c$ creates at least one new safe block at $((A'_c, A_{-c}), \succ)$ compared to $(A, \succ)$ by adding institutions to her acceptance set, then Lemma~\ref{L2} implies that $c$ would be assigned to an institution in this safe block at $((A'_c, A_{-c}), \succ)$, unless $c$ is assigned at a prior step at $((A'_c, A_{-c}), \succ)$. In the latter case, if this prior step is the same at profile $(A \succ)$, then $\mu_c \in A_c$, which is a contradiction. Thus, $c$ is assigned to one of the $k$ institutions in a new safe block $((A'_c, A_{-c}), \succ)$ compared to $(A, \succ)$. Note that since $c$ was not assigned to $\mu'_c$ at $(A,\succ)$, it is not possible that if $\mu'_c$ is earlier in the institution permutation than at $(A,\succ)$, while otherwise the matching procedure remains the same, that $c$ is assigned to $\mu'_c$, except if $c$ was not on the acceptance list of $\mu'_c$ at $(A, \succ)$, that is, $c$ can only be assigned to $\mu'_c$ in this case if $\mu'_c \in A'_c \setminus A_c$. This means that according to the true preferences of agent $c$, $\mu'_c$ is unacceptable, that is $\mu'_c \notin A_c$. This means that the manipulation attempt was not successful, which is a contradiction.         

\medskip 

\textbf{Completion of Case 2:} 

\smallskip

We have shown in subcase 2.1 that agent $c$ cannot manipulate the outcome by only subtracting institutions and remains unmatched. Let $A''_c = A_c \cap A'_c$. Then $ A''_c \subseteq A_c $ and thus $\varphi_c((A''_c, A_{-c}), \succ) = 0$, by subcase 2.1. However, $ A''_c \subseteq A'_c $ and since we have shown in subcase 2.2 that agent $c$ cannot manipulate the outcome by only adding institutions, $\varphi_c((A'_c, A_{-c}), \succ) \notin A''_c$. Since it is a feasible matching, $\varphi_c ((A'_c, A_{-c}),\succ)\notin (A_c \setminus A'_c)$, and therefore $\varphi_c ((A'_c, A_{-c}) \succ) \notin A_c$, which means that the manipulation attempt was not successful. Therefore, subcases 2.1 and 2.2 together cover all possible cases of manipulation by agent $c$ for Case~2. 

\bigskip

\noindent In sum, since an arbitrarily chosen agent $c$ cannot manipulate in either Case 1 or Case 2, a SAFE mechanism is strategyproof.
\end{proof}

\begin{theorem}
The SAFE/Rank-Maximal Mechanisms are strategyproof for institutions. 
\end{theorem}
\begin{proof}
The set of lexi-optimal set of institutions that can be matched cannot be changed by any institution, as it is based on the acceptability graph. 

Suppose for contradiction that SAFE/Rank-Maximal Mechanisms are not strategyproof for institutions. Then there is an institution $d$ that can change its preference over agents to be assigned a more preferred agent. Before $d$'s turn comes, note that $d$ cannot affect the matches of institutions before it in the order $\pi$ as the decision about whether a previous institution $d'$ can match to a particular agent depends on the acceptability graph. 
When $d$'s turn comes in Algorithm~\ref{mech:RM}, it is assigned the most preferred agent that it can get while ensuring that the remaining institutions in $W$ are also matched. Since $d$ cannot affect the acceptability graph, Algorithm~\ref{mech:RM} is strategyproof for institutions.

\end{proof}

We show that the outcome of the SAFE/Rank-Maximal Mechanisms can be computed in polynomial time. 

\begin{theorem}
The outcome of the SAFE/Rank-Maximal Mechanisms can be computed in polynomial time. 
\end{theorem}
\begin{proof}
Firstly, Algorithm~\ref{mech:RM} builds up a lexi-optimal set of institutions $W$ that can all be matched by repeatedly calling an algorithm that computed a maximum size matching of a given a graph.  A maximum size matching of a bipartite graph can be computed in polynomial time (more precisely, cubic time in the number of vertices of the graph) by any of several well-established methods such as Kuhn's algorithm \citep{Kuhn55a} or the Hopcroft-Karp-Karzanov algorithm~\citep{HoKa73a,Karz73a}.
While ensuring that a matching that matches each element of $W$ exists, it finds the highest priority agent $i$ of the highest ranked institution $c$ such that if $i$ and $c$ are matched, every element in $W$ can be matched. Hence, if $i$ and $c$ are removed from the problem, all the remaining elements in $W$ can still be matched. The process iterates to match each element in $W$ and returns a matching that  match each element of $W$.\end{proof}

\begin{theorem}
The SAFE/Rank-Maximal Mechanisms are Pareto-efficient for the institutions. 
\end{theorem}
\begin{proof}
Suppose that the outcome is not Pareto-efficent. Then it admits at least one trading cycle. Among all cycles identify the highest ranked institution $d$. At the point where $d$ is permanently matching with some agent $i$, it could have  been permanently matched with the higher priority agent it points to in the trading cycle which is a contradiction.
Hence, the outcome is Pareto-efficient which also implies fairness. 
\end{proof}

Next, we consider non-bossiness that has been studied in several allocation and matching contexts~\citep{Papa01b,Fuhi10a,Sven99a}.
A mechanism  is \textit{non-bossy} if no unmatched agent can change her acceptable set such that she remains unmatched but the matching changes.   
The  (agent-proposing) Deferred Acceptance mechanism when applied to our setting by using arbitrary tie-breaking in agents' preferences does not satisfy non-bossiness. 
The REV rule of Aziz and Brandl (2024) also violates non-bossiness~\citep{AzBr21a}.\footnote{In REV, it is even possible that an unmatched agent remains unmatched but manages to change its preferences and change the set of matched agents~\citep{AzBr21a}. }

In contrast, we show that our rules satisfy non-bossiness. 

\begin{theorem}
The SAFE/Rank-Maximal Mechanisms are non-bossy. 
\end{theorem}
\begin{proof}
Assume that agent $i$ is unmatched and reduces her acceptable set of institutions. 
The acceptability graph is the same except for some edges adjacent to $i$ which have been removed. This means that the lexi-optimal set of institutions does not change, as the previous lexi-optimal set of institutions can still be matched to the same agents as previously, whose acceptability edges are the same.  Consequently, in each step of the algorithm, each institution is matched to the same agent as previously. Hence, the matching does not change. Given the symmetry of this argument and the assumption that $i$ remains unmatched after changing her acceptable set, it also follows that if agent $i$ expands her acceptable set of institutions while remaining unmatched, the matching does not change. Finally, a combination of these two arguments (reducing and expanding the acceptable set) implies the result. 
\end{proof}

\section{Conclusion}

In a model with dichotomous preferences of agents and strict priorities of institutions we identify a set of mechanisms, the SAFE/Rank-Maximal mechanisms,  which always maximize the matching size and do not violate the preferences or priorities over the agents. The SAFE mechanisms are sequential and allow institutions to choose their highest-priority agent according to a specific sequence of institutions among all agents who find this institution acceptable. We also show that SAFE mechanisms are strategyproof, that is, agents have no incentive to try to manipulate it by misrepresenting their preferences.  


\bibliographystyle{natbib}


\section*{Appendix}

\subsection*{Proof of Lemma~\ref{L1}}
\begin{proof}
	Let an acceptance list be given for each institution in $D$.  Let a set $\tilde{D}\subseteq D$ of institutions be under-acceptable.  Suppose, by contradiction, that $\tilde{D}$ does not contain any safe block. Then there is no institution in $\tilde{D}$ with exactly one agent on its acceptance list. Thus, each institution in $\tilde{D}$ has at least two agents on its acceptance list. If there are two institutions  with two agents only on their acceptance lists jointly, then these two institutions constitute a safe $2$-block. Thus, each pair of institutions in $\tilde{D}$ is over-acceptable. Assume that $|\tilde{D}| > 2$. Let $k \ge 2$ be such that each set of less than or equal to $k$ institutions is over-acceptable within $\tilde{D}$. Since any set of $k$ institutions is over-acceptable, it follows that any set of $k+1$ institutions is either equal-acceptable or over-acceptable. If there exists a set of $k+1$ institutions which is equal-acceptable then it is a safe block, since all subsets of this set are over-acceptable. As $\tilde{D}$ contains no safe blocks, this is a contradiction.  Thus, each set of $k+1$ institutions is over-acceptable. By induction, $\tilde{D}$ is over-acceptable, which is a contradiction, since it is assumed to be under-acceptable. Therefore, each under-acceptable set of institutions contains a safe block. 
\end{proof}

\subsection*{Proof of Lemma~\ref{L3}}
\begin{proof}
	Let an acceptance list be given for each institution in $\bar{D}$, where $\bar{D} \subseteq D$. Let $D' \subseteq \bar{D} $  be a safe $k$-block in $\bar{D}$. Take any $k-1$ institutions from $D'$, say $D' \setminus \{d\}$, where $d \in D'$. The acceptance lists of institutions in $D' \setminus \{d\}$ contain at most $k$ agents collectively.  Consider the following three cases.
	Collectively there are a) less than $k-1$ agents, b) $k-1$ agents, and c) $k$ agents on the acceptance lists of the $k-1$ institutions in $D' \setminus \{d\}$. We will show that cases a) and b) lead to contradictions, and prove the statement in the lemma for case c).

	In case a) $D' \setminus \{d\}$ is under-acceptable, and thus Lemma~\ref{L1} implies that it contains a safe block. Since a safe block is equal-acceptable, any safe block in $D' \setminus \{d\}$ has to be a strict subset of $D' \setminus \{d\}$, and any safe block in $D' \setminus \{d\}$ is a safe block in $\bar{D}$. This contradicts the fact that $D' $ is a safe block in $\bar{D}$.

	In case b) $D' \setminus \{d\}$ is an equal-acceptable set and hence it is safe block in $\bar{D}$. This contradicts the fact that $D'$ is a safe block in $\bar{D}$.

	In case c), after assigning an agent to $d$ from $d$'s acceptance list and removing this agent from the acceptance lists in $D' \setminus \{d\}$, there are $k-1$ agents remaining collectively on the $k-1$ acceptance lists of institutions in $D' \setminus \{d\}$. Suppose by contradiction that it is not feasible to assign these remaining $k-1$ agents to the $k-1$ institutions in $D' \setminus \{d\}$. Then there exists a subset of $D' \setminus \{d\}$ which is an under-acceptable set, since otherwise the $k-1$ agents would be feasible to assign by Hall’s theorem. Then Lemma~\ref{L1} implies that there is a safe block in $D' \setminus \{d\}$, which is a contradiction since $D'$ is a safe block. Thus, it is feasible to assign the remaining $k-1$ agents to the $k-1$ institutions in $D' \setminus \{d\}$.

	Therefore, since $d$ is an arbitrary institution in $D'$, it is feasible to assign each agent on the acceptance list of at least one institution in $D'$ to an institution in $D'$ that is acceptable to this agent, regardless of the first institution that is assigned an agent from its acceptance list. Given that $D'$ was an arbitrary safe block in $\bar{D}$ where $\bar{D} \subseteq D$, the proof is completed. 
\end{proof}


\subsection*{Proof of Lemma~\ref{L4}}

\begin{proof}
	Let an acceptance list be given for each institution in $\bar{D}$, where $\bar{D} \subseteq D$ such that there is no safe block in $\bar{D}$. Since there is no safe block in $\bar{D}$, there is no equal-acceptable subset of institutions in  $\bar{D}$, and thus Lemma~\ref{L1} implies that each subset of institutions in $\bar{D}$ is over-acceptable. Let $d \in \bar{D}$ and assign an agent to $d$ from $d$'s acceptance list. After removing this agent from the acceptance lists in $D' \setminus \{d\}$, there are at least $k$ agents remaining collectively on the $k-1$ acceptance lists of institutions in $D' \setminus \{d\}$. Moreover, note that each non-empty subset of $D' \setminus \{d\}$ is either equal-acceptable or over-acceptable. Then Hall's theorem implies that it is feasible to assign each institution in $D' \setminus {d}$ an agent other than $d$ on its acceptance list, and the resulting matching is a maximum matching for  $D' \setminus {d}$. This means that, together with the assignment of an agent to $d$ from $d$'s acceptance list, we have a maximum matching for  $D'$.
\end{proof}

\begin{example}[DA is bossy]
	
	Consider the following preferences with the accompanying tie-breaking among the acceptable institutions. 
	
	\begin{align*}
		A_1&: (d_2, d_1) \\
		A_2&: (d_1, d_2)\\
		A_3&: (d_1, d_3)\\
		A_4&: (d_3)
	\end{align*}
	
	The acceptance lists of the institutions are as follows:
	
	\begin{align*}
		d_1&: 1, 3, 2\\
		d_2&: 2, 1\\
		d_3&: 4, 3
	\end{align*}
	
	The agent-optimal matching is  $d_1 - 1$, $d_2 - 2$, $d_3 - 4$
	Note that 3 is unmatched.
	
	If 3 reports $A'_3: (d_3)$ then the agent-optimal matching at $(A'_3, A_{-3})$ is  $d_1 - 2$, $d_2 - 1$, $d_3 - 4$, leaving 3 unmatched. 
	
\end{example}

\end{document}